\documentclass[letterpaper, 10 pt, conference]{ieeeconf}  

\IEEEoverridecommandlockouts                              

\usepackage{graphicx}
\usepackage{times} 
\usepackage{amsmath} 
\usepackage{amssymb}  
\usepackage{url}
\usepackage{cite}
\usepackage{algorithm,algorithmic}
\usepackage{cases}

\usepackage{theorem}
\theorembodyfont{\rmfamily}

\newtheorem{lem}{Lemma}
\newtheorem{defin}{Definition}
\newtheorem{theorem}{Theorem}
\newtheorem{prop}{Proposition}
\newtheorem{assum}{Assumption}

\usepackage{cleveref}
\crefname{assum}{assum}{assumptions}

\newcommand{\hs}{& \hspace{-2mm}}

\newcommand{\tm}{\theta_{\rm m}}
\newcommand{\htm}{\theta_{\rm m}}
\newcommand{\tb}{\theta_{\rm b}}
\newcommand{\htb}{\theta_{\rm b}}
\newcommand{\ssk}{s^{\rm s}_k}
\newcommand{\srk}{s^{\rm r}_k}
\newcommand{\sszk}{s^{\rm s}_{0:k}}

\newcommand{\sigszk}{\sigma^{\rm s}_{0:k}}
\newcommand{\cSsk}{\mathcal{S}^{\rm s}_k}
\newcommand{\cSrk}{\mathcal{S}^{\rm r}_k}
\newcommand{\iTheta}{{\it \Theta}}
\newcommand{\hiTheta}{{\it \Theta}}
\newcommand{\isk}{i^{\rm s}_k}
\newcommand{\irk}{i^{\rm r}_k}
\newcommand{\Isk}{I^{\rm s}_k}
\newcommand{\Irk}{I^{\rm r}_k}
\newcommand{\cIsk}{\mathcal{I}^{\rm s}_k}
\newcommand{\cIrk}{\mathcal{I}^{\rm r}_k}
\newcommand{\EUsk}{\bar{\mathbf{U}}^{\rm s}_k}
\newcommand{\EUrk}{\bar{\mathbf{U}}^{\rm r}_k}
\newcommand{\BRsk}{{\rm BR}^{\rm s}_k}
\newcommand{\BRrk}{{\rm BR}^{\rm r}_k}

\newcommand{\pia}{\hat{\pi}^{\ast}_{\rm m}}

\renewcommand{\QED}{\QEDopen}

\DeclareMathOperator*{\argmax}{arg\,max}

\title{\LARGE \bf
 Asymptotic Security of Control Systems by Covert Reaction:\\ Repeated Signaling Game with Undisclosed Belief
}

\author{Hampei Sasahara$^{1}$, Serkan Sar{\i}ta\c{s}$^{2}$, and Henrik Sandberg$^{1}$
\thanks{$^{1}$ Hampei Sasahara and Henrik Sandberg are with the Division of Decision and Control Systems,
	KTH Royal Institute of Technology,
	Stockholm, Sweden, {\tt\small \{hampei,hsan\}@kth.se}.}
\thanks{$^{2}$ Serkan Sar{\i}ta\c{s} is with the Division of Information Science and Engineering,
	KTH Royal Institute of Technology,
	Stockholm, Sweden, {\tt\small saritas@kth.se}.
	}
        }

\begin{document}

\maketitle
\thispagestyle{empty}
\pagestyle{empty}

\begin{abstract}

This study investigates the relationship between resilience of control systems to attacks and the information available to malicious attackers.
Specifically, it is shown that control systems are guaranteed to be secure in an asymptotic manner by rendering reactions against potentially harmful actions covert.
The behaviors of the attacker and the defender are analyzed through a repeated signaling game with an undisclosed belief under covert reactions.
In the typical setting of signaling games, reactions conducted by the defender are supposed to be public information and the measurability enables the attacker to accurately trace transitions of the defender's belief on existence of a malicious attacker.
In contrast, the belief in the game considered in this paper is undisclosed and hence common equilibrium concepts can no longer be employed for the analysis.
To surmount this difficulty, a novel framework for decision of reasonable strategies of the players in the game is introduced.
Based on the presented framework, it is revealed that any reasonable strategy chosen by a rational malicious attacker converges to the benign behavior as long as the reactions performed by the defender are unobservable to the attacker.
The result provides an explicit relationship between resilience and information, which indicates the importance of covertness of reactions for designing secure control systems.

\end{abstract}

\section{Introduction}
\label{sec:intro}

Modern cyber-physical systems are facing malicious actions by sophisticated adversaries because of increasing interconnectivity of the cyber and physical layers.
In fact, various incidents on cyber attacks on physical systems have been reported, which cover a broad range of applications such as an attack on an uranium enrichment plant, attacks on Ukrainian power plants, remote car hacking, and unmanned aerial vehicle hacking~\cite{Kushner2013Real,Lee2016Analysis,Miller2015Remote,Daniel2012Drone}.
Thus, analysis and synthesis of resilient cyber-physical systems have become a crucial topic in our community.

There are numerous studies discussing resilience of cyber-physical systems based on control theory, in which several definitions of resilience have been proposed thus far; see~\cite{Teixeira2015Secure,Mo2016On,Milosevic2017Analysis,Bai2017Data,Umsonst2017Security,Kung2017The,Chen2018Optimal,Guo2018Worst,Milosevic2019Estimating}.
In the existing studies, the control system to be defended is supposed to be protected with an attack-detection scheme.
Accordingly, the behavioral principle of malicious attackers is given as maximizing impacts on behavior of the control system through attacks without being detected.
In this framework, behaviors of the system, the attacker, and the defender after an alarm are not concerned.
An implicit premise is that effects of malicious attacks can immediately be eliminated or an excessive reaction, e.g., terminating the entire system, is performed once an alarm rings.


In actual systems, however, the post-alarm process should carefully be analyzed because root-cause analysis of anomalies is difficult and also suspension of whole parts of the system is undesirable due to false alarms and costs.
An approach to cope with attacks more flexibly is to choose a proper reaction against potentially executed harmful actions depending on the level of confidence on existence of a malicious attacker.
An example can be provided for networked systems, in which disconnecting some components that are presumed to be attacked is more acceptable than simply terminating the entire network when the operator's confidence is not firm.
The technical problem here is that, under this scenario sophisticated attackers choose their actions according to their measurement on the state of the system as well as defender's reactions.
Thus we can no longer employ the existing frameworks for resilience assessment because it is now required to handle interactions between attacker's actions and defender's reactions unlike the case where the attacker simply avoids being detected.

This study investigates the relationship between \emph{resilience} of control systems and \emph{information} available to malicious attackers under the setting above motivated by the following consideration.
Suppose that a malicious attacker can fully observe reactions executed by the defender.
When observing aggressive reactions, the attacker recognizes that the confidence on existence of a malicious attacker is fairly firm and then a reasonable behavior of the attacker is to stay calm.
Conversely, if a non-aggressive reaction is observed, the attacker is allowed to execute an offensive action.
This discussion leads to the expectation that leak of information on the reactions implemented by the defender causes serious damage to the system, and conversely, also that resilience is enhanced by rendering reactions covert.
A naturally arising question here is how, and how much, the resilience is strengthened by covertness of reactions.

In order to give an answer to the question, we analyze the situation using game theory.
In particular, we formulate behaviors of the decision makers as a repeated signaling game~\cite{Kaya2009Repeated}.
In a signaling game, there are two players, one of which, referred to as a sender, has private information called \emph{type} and chooses her action based on her preference depending on her type~\cite{Crawford1982Strategic,Kreps1994Chapter}.
The action has a role as signaling to the other player, referred to as a receiver, who forms a belief on the sender's type based on received signals and chooses a reasonable reaction under the belief.
A widely employed solution concept is perfect Bayesian Nash equilibria, on which both player's strategies are best responses to each other and the belief of the receiver is updated based on Bayes' rule.
Typically, the reasonable actions of the players are supposed to be given as one of the equilibria on the premise that the prior belief and its transition are public information.
In our case, however, the exact value of the belief is undisclosed to the attacker owing to covertness of reactions and hence her reasonable strategies cannot be determined using the usual concept.
Because the accurate value of the belief is unknown, the attacker has to estimate the belief, i.e., ``belief of belief'' is required for determining the attacker's reasonable actions.
To handle this situation, we introduce a novel framework, under which the estimated belief is specifically given as the conditional expected value of the true belief.
Based on the framework, we can determine well-posed reasonable attacker's actions.

Our analysis under the presented framework reveals that \emph{asymptotic security} is achieved by covert reaction for general control systems.
Specifically, it is shown that any reasonable strategy chosen by a rational malicious attacker converges to the benign behavior as time proceeds.
The finding is based on the fact that the estimated belief from the attacker is a monotonically non-decreasing function whatever strategy the attacker chooses.
From the monotonicity, the estimated belief is eventually increased as the game proceeds.
When the belief is sufficiently large the reasonable action is restricted to the most unsuspicious behavior, which is given by the benign behavior.
The obtained result gives an explicit relationship between resilience and information and suggests that covertness of reactions is an essential requirement for designing secure control systems.

Existing literature considering control systems under Bayesian games where some players do not have exact values of other players' beliefs can be found~\cite{Gupta2014Common,Nayyar2014Common,Vasal2019A,Heydaribeni2019System}.
Typically, the situation emerges with \emph{asymmetric information}, i.e., some information of the game is assumed to be private.
The approach for handling the situation in~\cite{Gupta2014Common,Nayyar2014Common,Vasal2019A} is based on the idea that the belief of all players are constructed only with common information.
It is shown in~\cite{Heydaribeni2019System} that if the game is a linear quadratic game then the true belief is constructed only with each player's local information.
Analysis of system security based on signaling game is studied in~\cite{Pawlick2019Modeling}, the modeling in which is similar to ours.
The main difference compared to this study is that the belief is assumed to be known to both players and a single stage game is considered in~\cite{Pawlick2019Modeling}.
In economics literature, games without knowledge of actions performed by players are categorized into a Bayesian game with imperfect information, in which the solution concept is given as sequential equilibria~\cite{Kreps1982Sequential,Fudenberg1991Perfect,Cole2001Dynamic}.
Our solution concept differs from sequential equilibria from the perspective that our concept is given with the sender's belief about the receiver's belief while sequential equilibria are given with the sender's belief about all variables in the game on each information set.
Furthermore, the notion of hierarchies of beliefs is proposed for treating the ``belief of belief'' concept in the general setup~\cite{Mertens1985Formulation,Borgers2005Introduction}.

The organization of this paper is as follows.
In Section~\ref{sec:sys_game}, we first show a motivating example for our analysis.
Then we give the model of the system to be defended and behaviors of a malicious attacker and the defender based on a repeated signaling game formulation.
Section~\ref{sec:ana} provides an analysis of the formulated game and reveals that asymptotic security is achieved by covert reaction.
The theoretical findings are verified through a numerical example in Section~\ref{sec:num} and Section~\ref{sec:conc} draws conclusion.

\subsection*{Notation}
Let $\mathbb{Z}_+$ and $\mathbb{R}$ denote the sets of nonnegative integers and real numbers, respectively.
Let $\mathfrak{P}(\mathcal{X})$ denote the power set of a set $\mathcal{X}$.
The probability space considered in this paper is denoted by $(\Omega,\mathfrak{F},{\rm Prob})$.
We use upper case letters for random variables and lower case letters for their realizations.
The expected value of $X$ is denoted by $\mathbb{E}[X]$.
For a random variable $X:\Omega\to \mathcal{X}$, we use $X\in \mathcal{X}$ for representing its codomain.
For discrete random variables $X$ and $Y$, let $p(x)$ and $p(y|x)$ denote the probability mass function of $X$ and the conditional probability mass function of $Y$ given $X=x$, respectively.
Similarly, $\mathbb{E}[Y|x]$ denotes the conditional expected value of $Y$ given $X=x$.
Let $x_{0:k}$ denote a tuple $(x_0,\ldots,x_k)$ of the sequence $x_0,\ldots,x_k$.

\section{Control System with Covert Reaction}
\label{sec:sys_game}

\subsection{Motivating Example}
\label{subsec:mot}

Consider the networked system shown in Fig.~\ref{fig:ex}.
The upstream process and the downstream process are interconnected through a network over an overt channel.
The overt channel is possibly under man-in-the-middle attack where the adversary can receive and alter the signal on the channel.
To cope with the attack, the system operator equips a redundant secure channel, called covert channel~\cite{Lampson1973A}, through which secret information can be transferred being unrecognized by other third parties.
A covert channel is typically regarded as a vulnerability of communication~\cite{Zander2007A} but can proactively be used for enhancing reliability~\cite{Kaur2016Covert}.
Although the subsystems can securely communicate through the covert channel, the overt channel is preferable to use if no attacks are executed since communication over the covert channel requires an unusual protocol and would cause a high cost or performance degradation.
Moreover, the operator utilizes the model-based anomaly detection scheme~\cite{Ding2013Model} with which the defender can observe the residual obtained as the difference between the actual signal through the overt channel and the ideal signal calculated through a model of the system.
Basically, the residual is not always zero owing to noise but is small without attacks and hence large residual encourages the operator to use the covert channel.
The objective of the operator is to achieve an efficient operation by properly choosing the channel used for the communication at each time step based on the measured residual.
Conversely, the objective of the adversary is to disturb the downstream process by choosing proper signals on the overt channel.

\begin{figure}[t]
\centering
\includegraphics[width = .95\linewidth]{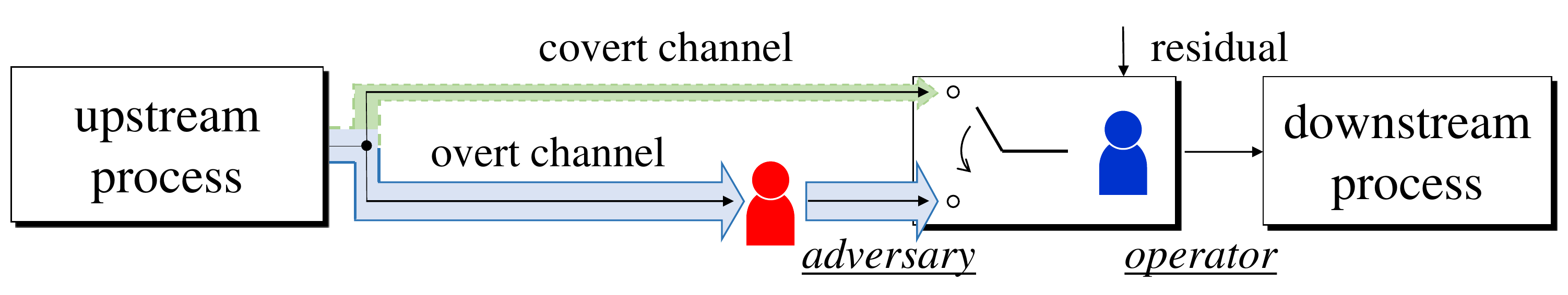}
\caption{Example: A networked system with an overt channel under man-in-the-middle attack and a secure covert channel that would cause a costly communication or performance degradation.}
\label{fig:ex}
\end{figure}

We now explain the idea of  ``covert reaction'' with this example.
Let us suppose that no data is transferred through the overt channel when the covert channel is used.
Then the adversary can recognize the channel used for each communication, which corresponds to the case of \emph{overt reaction.}
If usage of the overt channel is observed, which indicates that presence of the attacker is not suspected, then large disturbance can be injected.
If no data through the over channel is observed, the attacker tries to lurk by stopping injection of attacks.
Hence, in the case of overt reaction, the adversary can choose a reasonable action according to the observed scenario.
A possible approach that inhibits this strategy is to transfer data through both the channels when the covert channel is used for actual communication.
Then, because the overt channel is used in both cases, the true scenario cannot be identified from the attacker's perspective, which corresponds to the case of \emph{covert reaction.}
Thus, in the case of covert reaction, the adversary has to take an alternative strategy, under which the damage on the system would be less harmful.
From this discussion, we expect that covert reaction can enhance resilience of control systems.

This example motivates us to thoroughly analyze resilience of control systems with covert reaction.
In the next subsection, we provide a model of the control system to be defended for mathematical analysis.

\subsection{System Description}
\label{subsec:sys}

Consider a control system being possibly under attack shown in Fig.~\ref{fig:sys}.
Let $U_k \in \mathcal{U}$ and $X_k \in \mathcal{X}$ denote the input signal and the state of the system at the time step $k \in \mathbb{Z}_+$, respectively.
The input signals are assumed to be mutually independent, i.e., $U_i$ and $U_j$ are independent for any $i,j\in \mathbb{Z}_+$.
There is a party, called a \emph{sender}, who can alter the behavior of the system through an \emph{action} $a_k \in \mathcal{A}$.
Let $\theta \in \iTheta=\{\tb,\tm\}$ denote the \emph{type} of the sender,
where the symbols $\tb$ and $\tm$ imply that the sender is benign or malicious, respectively.
The case $\theta=\tb$ corresponds to the control system under normal operation where effects of the action $a_k$ can be ignored, e.g., $a_k=0$ if $\mathcal{A}$ is the Euclidean space.
We denote the map of the system by $\Sigma_k: \mathcal{U}^{k+1} \times \mathcal{A} \to \mathcal{X}$.
Then the $k$th state is determined according to the relationship
$x_k = \Sigma_k(u_{0:k},a_k)$
with given input signals and an action where effects of the initial value of the state are ignored because the analysis in this paper is independent of the initial value.
Note that we assume that the state is affected only by the instantaneous action.
This property is favorable to the attacker in terms of avoidance of detection because there are no lingering effects of an attack.
In this sense, we consider a severe situation for defending the control system.

\begin{figure}[t]
\centering
\includegraphics[width = .90\linewidth]{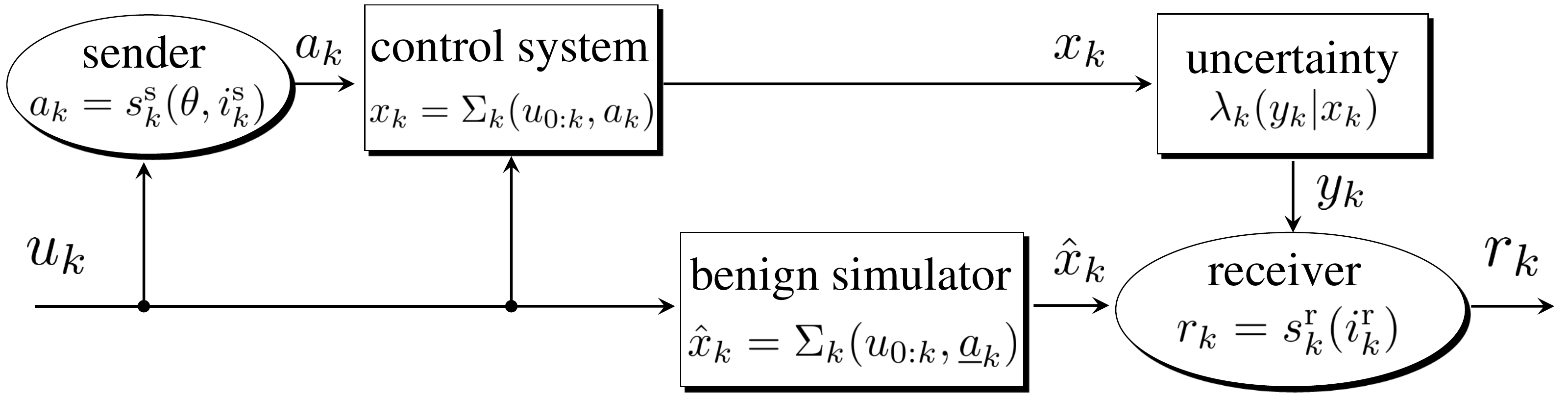}
\caption{Signal-flow diagram of the system under potentially malicious actions and the defense architecture for detection and reaction.}
\label{fig:sys}
\end{figure}

To detect the existence of a malicious sender and cope with potentially malicious actions, we consider a commonly used detection and reaction architecture shown in Fig.~\ref{fig:sys}.
The basic idea of the defense mechanism is to compare the actual signal and the ideal signal under the normal operation and to perform a proper reaction based on the result of the detection process.
Let $Y_k \in \mathcal{Y}$ denote the measurement output which is obtained through an uncertain environment, e.g., noisy measurement and modeling error.
The characteristic of the uncertainty is modeled by a memoryless channel $\lambda_k(y_k|x_k):= p(y_k|x_k) = p(y_k|x_{0:k})$.
The memoryless property is also favorable to the attacker for the same reason as the one of the control system.
Let $\hat{X}_k \in \mathcal{X}$ denote the ideal state at the $k$th step under the benign operation generated through the benign simulator
$\hat{x}_k = \Sigma_k(u_{0:k},\underline{a}_k)$
where $\underline{a}_k$ is the action that is performed if the sender is benign.
Note that $\underline{a}_k$ is known to the players owing to Assumption~\ref{assum:benign} shown below.
Based on the measurement, the other party, called a \emph{receiver}, chooses an action $r_k \in \mathcal{R}$ at each time step.
We henceforth refer to $r_k$ as a \emph{reaction} for emphasizing that $r_k$ denotes a counteraction against potentially malicious attacks.

For avoiding trivial discussions, we make an assumption that guarantees validity of the defense architecture.
\begin{assum}[Valid Defense Architecture]\label{assum:defense}
The defense architecture is valid, i.e.,
\begin{enumerate}
\item For any $u_{0:k}$, if $x_k=\hat{x}_k$ then $a_k=\underline{a}_k$.
\item Let $I(X_k;Y_k)$ denote the mutual information of $X_k$ and $Y_k$~\cite{Cover2006Elements}.
Then
$\inf_{k\in \mathbb{Z}_+}I(X_k;Y_k)> 0$
holds.
\end{enumerate}
\end{assum}
The first assumption means that the system is input observable~\cite{Hou1998Input}, i.e., any malicious attack is immediately noticeable if there is no uncertainty in the measurement process.
Owing to this assumption, we disregard attacks that are impossible to detect even with the most powerful detection tool, such as covert attacks~\cite{Smith2015Covert} and zero-dynamics attacks~\cite{Pasqualetti2015Control}.
The second assumption simply states that the measurement signal has nonzero information on the state,
which means that the channel is valid for attack detection including the limit.

Moreover, for simplicity, the following technical assumption is made.
\begin{assum}[Discrete Variables]\label{assum:tech}
The sets of signals and actions, namely, $\mathcal{U},\mathcal{X},\mathcal{Y},\mathcal{A},\mathcal{R}$ are finite sets.
\end{assum}
Assumption~\ref{assum:tech} makes all random variables discrete.
Then it suffices to consider only probability mass functions.

\subsection{Game Description}
\label{subsec:game}

Based on the system description, we provide each player's decision policy about her actions through a repeated game formulation.
For the formulation, the notions of \emph{strategy, utility, belief,} and \emph{best response} are needed.

The sender chooses an action at each time step based on the type and available information.
Let $\ssk \in \cSsk$ denote the \emph{strategy} of the sender at the $k$th step.
Suppose that the sender uses \emph{pure strategies}.
Then the space of the $k$th sender's strategies is given by
$\cSsk= \{ \ssk:\iTheta \times \cIsk \to \mathcal{A}\}$
where $\cIsk$ denotes the set of the $k$th sender's available information, which consists of signals measured by the sender at the $k$th step.
The elements of $\cIsk$, referred to as information sets in game theory, are explicitly determined in the next subsection (see Assumption~\ref{assum:info}).
Similarly, let $\srk \in \cSrk$ denote the strategy of the receiver at the $k$th step.
The receiver uses pure strategies as well and then the space of the $k$th receiver's strategies is given by
$\cSrk=\{\srk: \cIrk \to \mathcal{R}\}$
where $\cIrk$ denotes the set of the $k$th receiver's available information.

Utilities are employed for representing players' preferences.
Let $\mathbf{U}^{\rm s}_k:\iTheta\times \mathcal{X}\times \mathcal{A}\times \mathcal{R}\to \mathbb{R}$ denote an immediate utility function for the sender such that $\mathbf{U}_k^{\rm s}(\theta,x_k,a_k,r_k)$ gives the utility that she receives when her type is $\theta$, the sender plays the action $a_k$, and the receiver chooses the reaction $r_k$.
Similarly, let $\mathbf{U}^{\rm r}_k: \iTheta\times\mathcal{X}\times \mathcal{A} \times \mathcal{R}\to \mathbb{R}$ denote the receiver's utility function such that $\mathbf{U}^{\rm r}_k(\theta,x_k,a_k,r_k)$ gives her payoff under the same scenario.

As a preparation for analysis, we put an assumption on preferences of the benign sender.
We suppose that the system is optimally controlled without any specific actions of the sender and then actions plays the role only to disturb the system.
Since the benign sender has no motivation to disturb the system, doing nothing becomes the optimal action.
The following assumption is made.
\begin{assum}[Benign Sender's Preference]\label{assum:benign}
There exists
$a_{\rm b} \in \mathcal{A}$ such that $\mathbf{U}^{\rm s}_k(\tb,x_{\rm b},a_{\rm b},r) > \mathbf{U}^{\rm s}_k(\tb,x,a,r')$
where $x_{\rm b} := \Sigma_k(u_{0:k},a_{\rm b})$ and $x := \Sigma_k(u_{0:k},a)$ for any $a\in \mathcal{A}\setminus\{a_{\rm b}\}, u_{0:k}\in \mathcal{U}^{k+1}, r,r'\in \mathcal{R},k\in\mathbb{Z}_+$.
\end{assum}
In Assumption~\ref{assum:benign}, the benign action $a_{\rm b}$ means that the sender does not perform any particular active operations.
We do not put any assumptions on the utilities of the malicious sender and the receiver at this moment.
Instead, we make an alternative assumption that dominates their preferences in the next section (see Assumption~\ref{assum:prefe} in Section~\ref{sec:ana}).

When some variables of the utilities are unavailable to the players, \emph{expected utilities} are used for choosing a reasonable strategy instead of the utility itself.
Let $\EUsk:\iTheta\times \mathcal{A}\times \cIsk\times \cSrk\to \mathbb{R}$ denote the sender's expected utility at the $k$th step under a fixed sender's strategy with a realization of her available information for a given receiver's strategy such that
$\EUsk(\theta,a_k,\isk,\srk) := \mathbb{E}\left[\mathbf{U}^{\rm s}_k\left(\theta,X_k,a_k,\srk(\Irk)\right)|\isk\right].$

Similarly, we define the receiver's expected utility.
However, since $\theta$ is not a random variable, the expected value cannot be taken in the normal sense.
Thus, the receiver prepares a \emph{belief} $\pi_k \in \Pi$ such that $\pi_k:\iTheta\to (0,1)$ represents the possibility of the type of the sender for each time step from the viewpoint of the receiver.
Let $\EUrk: \mathcal{R} \times \cIrk \times \cSsk \times \Pi \times \mathcal{S}^{\rm s}_{0:k} \to \mathbb{R}$ denote the receiver's expected utility at the $k$th step under the belief such that
$\EUrk(r_k,\irk,\ssk,\pi_k,\sigszk) := \mathbb{E}\left[ \mathbf{U}^{\rm r}_k(\theta_k,X_k,\ssk(\theta_k,\Isk),r_k ) | \irk\right]$
where $\theta_k\in \hiTheta$ is a random variable obeying the probability distribution obtained by the belief $\pi_{k+1}$ in~\eqref{eq:pik}, shown below, and $\sigszk \in \mathcal{S}^{\rm s}_{0:k}$ is a profile of sender's strategies used for calculation of $\pi_{k+1}$.
Note that we use the symbol $\sigszk$ instead of $s^{\rm s}_{0:k}$ for emphasizing that the strategies $\sigszk$ are not the strategies that are actually chosen by the sender but the estimated strategies from the perspective of the receiver.

The receiver updates her belief based on her observation as the time proceeds according to Bayes' rule, which is given by
${\rm Prob}(\mathcal{E}_1|\mathcal{E}_2) = {\rm Prob}(\mathcal{E}_1\cap \mathcal{E}_2)/{\rm Prob}(\mathcal{E}_2) = {\rm Prob}(\mathcal{E}_2|\mathcal{E}_1){\rm Prob}(\mathcal{E}_1)/{\rm Prob}(\mathcal{E}_2)$
for two events $\mathcal{E}_1$ and $\mathcal{E}_2$.
The dynamics of the belief obeys
\begin{equation}\label{eq:pik}
 \pi_{k+1} = f(\pi_k,\irk,\sigszk)
\end{equation}
with an initial belief $\pi_0 \in \Pi$ such that $\pi_0(\htb)+\pi_0(\htm)=1$ and $\pi_0(\htb),\pi_0(\htm)\in (0,1)$ where $f:\Pi\times \cIrk \times \mathcal{S}_{0:k}^{\rm s} \to \Pi$ represents Bayes' rule under a given sender's strategy given by
\begin{equation}\label{eq:Bayes}
 f\left(\pi,\irk,\sigszk\right)(\theta) := \dfrac{ p_{\sigszk}\left(\irk|\theta\right)\pi(\theta) }{\displaystyle{\sum_{\theta' \in \hiTheta} p_{\sigszk}\left(\irk|\theta'\right)\pi(\theta')}}
\end{equation}
where $\sigszk \in \mathcal{S}_{0:k}^{\rm s}$ represents a profile of estimated sender's strategies used for calculating the posterior belief and $p_{\sigszk}(\cdot)$ represents the corresponding probability under $A_i=\sigma^{\rm s}_i(\theta,I^{\rm s}_i)$ for $i=0,\ldots,k$.

A reasonable choice of each player's strategies is given by a \emph{best response} provided that the opponent's strategy is given.
A sender's strategy $\ssk$ is said to be a best response for a given receiver's strategy $\srk$ under the type $\theta$ when $\ssk$ satisfies the condition
$\ssk(\theta,\isk) \in \argmax_{a_k \in \mathcal{A}} \EUsk(\theta,a_k,\isk,\srk)$
for any $\isk \in \cIsk$, where it is supposed that the best response depends only on the immediate utility for simplicity.
To ease notation, we let $\BRsk:\iTheta\times\cSrk\to\mathfrak{P}(\mathcal{S}^{\rm s}_k)$ denote the set-valued function of all sender's best responses for a given receiver's strategy.
Similarly, a receiver's strategy $\srk$ is said to be a best response for a given $\ssk$ with the estimated sender's strategies $\sigszk$ under the belief $\pi_k$ when $\srk$ satisfies the condition
$\srk(\irk) \in \argmax_{r_k \in \mathcal{R}} \EUrk(r_k,\irk,\ssk,\pi_k,\sigszk)$
for any $\irk \in \cIrk$.
Let $\BRrk:\cSsk\times\Pi\times\mathcal{S}_{0:k}^{\rm s} \to\mathfrak{P}\left(\cSrk\right)$ denote the set-valued function of all receiver's best responses.

Finally, all models of the game and the initial belief are assumed to be known to both players.

Basically, the players choose their own strategies along the path generated by the best response.
The best response of the malicious sender is influenced by the receiver's belief through the receiver's strategy.
Thus the malicious sender is allowed to execute aggressive actions when she recognizes that the belief $\pi_k(\htm)$ is small.
This discussion leads us to the expectation that, the more information on the belief is leaked, the more serious the damage caused by the attack is, and vice versa.
From this perspective, we analyze the behaviors of the players under \emph{covert reactions} formulated in the subsequent subsection.

\subsection{Best Response Strategies under Undisclosed Belief}
\label{subsec:equ}

Consider the case where the receiver successfully conceals information on the belief from the sender, namely, the case of covert reaction.
The following assumption is made.
\begin{assum}[Covert Reaction]\label{assum:info}
The available information of the sender at the $k$th step is given by
$\isk=\{u_{0:k},a_{0:k-1}\}.$
The available information of the receiver is given by
$\irk = \{u_{0:k},\hat{x}_{0:k},y_{0:k},r_{0:k-1}\}.$
\end{assum}
Note that $\isk$ does not include $y_{0:k}$ and $r_{0:k-1}$.
If $r_{0:k-1}$ is available to the sender, the region of the belief can be computed and hence $r_{0:k-1}$ should be covert.
Further, since $\hat{x}_k$ can be computed from $u_{0:k}$, for keeping the belief undisclosed the receiver necessarily conceals $y_{0:k}$, with which the transition of the belief can completely be calculated.

One difficulty for the analysis of the game under Assumption~\ref{assum:info} is that a commonly used equilibrium concept, such as perfect Bayesian Nash equilibrium, cannot be employed because the players are unable to compute the equilibria for deciding reasonable actions.
To handle this situation, we introduce a novel framework in the following discussion.

From the perspective of the sender, one approach for choosing a reasonable strategy without the exact value of the belief is to estimate the belief only with available information.
We suppose that the sender uses the conditional expected value given a sender's strategy and the available information as an estimation of the belief.
Let $\hat{\pi}_k\in \Pi$ denote the estimated belief given by
$\hat{\pi}_{k+1}= {\rm EX}^{\pi}_k(\theta,s^{\rm s}_{0:k},i^{\rm s}_{k},s^{\rm r}_{0:k},\sigma^{\rm s}_{0:k})$
with ${\rm EX}^{\pi}_k:\iTheta\times\mathcal{S}^{\rm s}_{0:k}\times \mathcal{I}^{\rm s}_{k} \times \mathcal{S}^{\rm r}_{0:k} \times \mathcal{S}^{\rm s}_{0:k}\to \Pi$ defined by
${\rm EX}^{\pi}_k(\theta,s^{\rm s}_{0:k},i^{\rm s}_{k},s^{\rm r}_{0:k},\sigma^{\rm s}_{0:k})(\theta') := \mathbb{E}\left[ \pi_{k+1}(\theta')|\theta, s^{\rm s}_{0:k}, i^{\rm s}_{k},s^{\rm r}_{0:k}, \sigma^{\rm s}_{0:k}\right]$
where $\sigma^{\rm s}_{0:k}$ is the estimated sender's strategy profile from the receiver's side used for the Bayesian inference~\eqref{eq:Bayes}.
With the estimation, the sender can choose $s^{{\rm s}\ast}_k$ according to
\begin{equation}\label{eq:equili}
 \left\{
 \begin{array}{cl}
  s^{{\rm s}\ast}_k \hs \in \BRsk(\theta,s^{{\rm r}\ast}_k)\\
  s^{{\rm r}\ast}_k \hs \in \BRrk(s^{{\rm s}\ast}_k,\hat{\pi}_k,\sigszk)\\
  \hat{\pi}_k\hs = {\rm EX}^{\pi}_{k-1}(\theta,s^{\rm s}_{0:k-1},i^{\rm s}_{k-1},s^{{\rm r}\ast}_{0:k-1}, \sigma^{\rm s}_{0:k-1})
 \end{array}
 \right.
\end{equation}
provided that $\sigma^{\rm s}_{0:k}$ is known.

Since the exact values of $\sigszk$ are private information of the receiver, however, the sender has to estimate the estimated sender's strategies to obtain a reasonable strategy as well.
Let $\hat{\sigma}_{0:k}^{\rm s}$ denote the estimation of the estimated strategies $\sigszk$.
Because sender's strategies must be chosen to be reasonable, the conditions
\[
\left\{
 \begin{array}{cl}
  \hat{\sigma}^{{\rm s}}_i \hs \in {\rm BR}^{\rm s}_i(\theta,\hat{\sigma}^{{\rm r}}_i)\\
  \hat{\sigma}^{{\rm r}}_i \hs \in {\rm BR}^{\rm r}_i(\hat{\sigma}^{{\rm s}}_i,\hat{\hat{\pi}}_i,\hat{\hat{\sigma}}_{0:i}^{\rm s})\\
  \hat{\hat{\pi}}_i \hs = {\rm EX}^{\pi}_{i-1}(\theta,\hat{\sigma}^{{\rm s}}_{i-1},i^{\rm s}_{i-1},\hat{\sigma}^{\rm r}_{0:i-1},\hat{\hat{\sigma}}_{0:i-1}^{\rm s})
 \end{array}
\right.\\
\]
for any $i=0,\ldots,k$ are expected to hold under another estimated sender's strategy $\hat{\hat{\sigma}}_{0:k}^{\rm s}$.
Likewise, further estimations are required over and over and the same procedure is performed infinitely many times.

In this paper, we suppose that the sender uses the true strategies $s^{\rm s}_{0:k}$ as the estimated strategies $\sigszk$ for truncating the infinite chain of estimations.
Then we can define reasonable sender's strategies based on best responses as follows.
\begin{defin}[Best Response Strategy]\label{def:adm}
A sender's strategy profile $\{\ssk\}_{k\in \mathbb{Z}+}$ is said to be \emph{a sender's best response strategy profile} when there exists $\{s^{\rm r}_k\}_{k\in \mathbb{Z}_+}$ such that
\[
 \left\{
 \begin{array}{cl}
  s^{{\rm s}}_k \hs \in \BRsk(\theta,s^{{\rm r}}_k)\\
  s^{{\rm r}}_k \hs \in \BRrk(s^{{\rm s}}_k,\hat{\pi}_k,s^{\rm s}_{0:k})\\
  \hat{\pi}_k\hs = {\rm EX}^{\pi}_{k-1}(\theta,s^{{\rm s}}_{0:k-1},i^{\rm s}_{k-1},s^{\rm r}_{0:k-1},s^{\rm s}_{0:k-1})
 \end{array}
 \right.
\]
holds for any $k\in\mathbb{Z}_+, \theta\in \iTheta, i^{\rm s}_{k-1}\in \mathcal{I}^{\rm s}_{k-1}$.
\end{defin}
Definition~\ref{def:adm} is obtained by simply substituting $s^{\rm s}_{0:k}$ to $\sigszk$ in~\eqref{eq:equili}.
Note that, although we can define receiver's best response strategies in a similar way, we do not explicitly state the definition because it is unnecessary for our analysis.

The main objective of this study is to investigate the benefit induced by the confidentiality of the reactions on the resilience of the control systems.
More precisely, we aim at analyzing the reasonable behavior of the malicious sender under covert reactions.
Based on the preparation, we mathematically analyze sender's best response strategies in the next section.

\section{Analysis: Asymptotic Security}
\label{sec:ana}

In this section, we reveal that the control system can be guaranteed to be secure in an asymptotic manner through covert reaction, which results in \emph{asymptotic security.}

\subsection{Definition of Asymptotic Security}
First of all, we claim that the benign sender always chooses $a_{\rm b}$ under any best response strategies at any time.
\begin{prop}[Benign Sender's Strategy]\label{prop:benign}
Under Assumption~\ref{assum:benign}, any sender's best response strategy profile satisfies
$s^{\rm s}_k(\tb,I^{\rm s}_k) = a_{\rm b}$
almost surely for any $k\in \mathbb{Z}_+$.
\end{prop}
\begin{proof}
Assumption~\ref{assum:benign} directly leads to the claim.
\end{proof}


To clearly state the resilience achieved through covert reaction, we define a notion of security as follows.
\begin{defin}[Asymptotic Security]\label{def:asy}
The control system is said to be \emph{asymptotically secure} when any sender's best response strategy profile satisfies
$\ssk(\tm,I^{\rm s}_k)\to a_{\rm b}$
almost surely as $k \to \infty$.
\end{defin}
Definition~\ref{def:asy} means that the reasonable strategy of the malicious sender converges to the benign one regardless of her available information.
Note that in Definition~\ref{def:asy} the limit in $\mathcal{A}$ is taken in the following sense:
A sequence of actions $a_k \in \mathcal{A}$ for ${k\in \mathbb{Z}_+}$ converges to $a \in \mathcal{A}$ when there exists $N\in \mathbb{Z}_+$ such that $a_k=a$ for any $k>N$, which is obtained by giving the discrete topology~\cite{Munkres2000Topology} on $\mathcal{A}$. 

\subsection{Asymptotic Security by Covert Reaction}
While the preference of the benign sender is determined in Assumption~\ref{assum:benign}, we have made no assumptions on preferences of the malicious sender and the receiver.
We here make a reasonable assumption that should be satisfied by sender's best response strategies instead of making an assumption on their utilities themselves.
\begin{assum}[Players' Preference] \label{assum:prefe}
Any sender's best response strategy profile satisfies
${\rm Prob}\left(\mathcal{E}_{\{\ssk\}}\right)=1$
where
\[
 \mathcal{E}_{\{\ssk\}} :=\left\{\omega: \limsup_{k\in \mathbb{Z}_+}\pi_{{\rm m},k}(\omega) < 1\ {\rm under}\ \{\ssk\}\right\},
\]
which belongs to $\mathfrak{F}$, with $\pi_{{\rm m},k}:= \pi_k(\htm)$.

\end{assum}
Assumption~\ref{assum:prefe} means that any possible belief does \emph{not} approach one if the sender's strategy is given by best responses.
In other words, rational malicious senders avoid the situation where the receiver has a firm belief on existence of the malicious sender.
This assumption can be regarded as a replacement of an assumption on the utilities of the malicious sender and the receiver for determining their preferences.
Although it is more desirable to derive the statement in Assumption~\ref{assum:prefe} from an appropriate assumption on the utilities,
we treat the claim as an assumption in this paper.

Based on the preparation, the following theorem, the main result of this paper, holds.
\begin{theorem}[Asymptotic Security]\label{thm:main}
Let \Cref{assum:defense,assum:benign,assum:info,assum:prefe,assum:tech} hold.
Then the control system is asymptotically secure.
\end{theorem}
Essentially, Theorem~\ref{thm:main} is derived from the monotonicity of the estimated belief shown in the following lemma.
\begin{lem}[Monotonicity of Estimated Belief]\label{lem:mono}
Assume that $\{\ssk\}_{k\in \mathbb{Z}_+}$ is a sender's best response strategy profile.
Then under  \Cref{assum:defense,assum:benign,assum:info,assum:prefe,assum:tech}
\[
 \hat{\pi}_{k+1} (\htm) \geq \hat{\pi}_k(\htm)
\]
almost surely, where the equality holds if and only if $\ssk(\tm,\isk)=a_{\rm b}$.
\end{lem}
For the proofs of Lemma~\ref{lem:mono} and Theorem~\ref{thm:main}, see Appendix.

Lemma~\ref{lem:mono} implies that the estimated belief on existence of a malicious sender does \emph{not} decrease regardless of what strategy the malicious sender chooses at the $k$th step.
Furthermore, the estimated belief is invariant only if she chooses the benign action, which is the most unsuspicious behavior.
Because of the monotonicity of the estimated belief, strategies that can be chosen by the malicious sender are eventually restricted as the time proceeds.
When the expected belief becomes sufficiently large, the possible strategy becomes the benign one alone, which is the claim of Theorem~\ref{thm:main}.

Theorem~\ref{thm:main} implies that the control system is guaranteed to be secure in an asymptotic manner if the receiver's belief is kept to be undisclosed through covert reactions.
The result clarifies the relationship between secrecy of the defender's belief and the resilience of the control system to be defended.
What should be emphasized here is that no particular assumptions are made on the structure of the control system and the sender's utility function except for Assumption~\ref{assum:prefe}, under which the sender avoids having the receiver have a firm belief on existence of a malicious attacker.
In other words, the finding is irrelevant to the control system and the objective of the attacker.
In this sense, the result implies that covertness of reactions is an essential requirement for enhancing resilience of control systems.

\section{Numerical Example}
\label{sec:num}

We confirm the validity of the theoretical results through a simple numerical example.
Suppose that all sets of signals and actions are binary, i.e., $\mathcal{X}=\mathcal{Y}=\mathcal{U}=\mathcal{A}=\mathcal{R}=\{0,1\}$.
The input signals $U_k$ are supposed to be independent and identically distributed with ${\rm prob}(U_k=0)={\rm prob}(U_k=1)=0.5$.
The system is given by
$\Sigma_k(u_{0:k},0)=u_k$ and $\Sigma_k(u_{0:k},1)=1-u_k$.
The channel characteristic is given by $\lambda_k(0|0)=\lambda_k(1|1)=\lambda$ with a given $\lambda \in (0,1)\setminus{\{0.5\}}$.
The utility of the benign sender is given by
$\mathbf{U}^{\rm s}_k(\tb,x,0,r)=1$ and $\mathbf{U}^{\rm s}_k(\tb,x,1,r)=0$.
The other utilities are given by
\[
\begin{array}{ll}
 \mathbf{U}^{\rm s}_k(\tm,x,a,r)\hs=
 \left\{
 \begin{array}{cl}
 3,\hs {\rm if}\ a=1,r=0,\\
 2,\hs {\rm if}\ a=0,r=0,\\
 1,\hs {\rm if}\ a=0,r=1,\\
 0,\hs {\rm if}\ a=1,r=1,
 \end{array}\right.\\
 \mathbf{U}^{\rm r}_k(\theta,x,a,r)\hs=
 \left\{
 \begin{array}{cl}
 1,\hs {\rm if}\ \theta=\tb,r=0, \\
 \alpha,\hs {\rm if}\ \theta=\tm,r=1, \\
 0,\hs {\rm otherwise},
 \end{array}\right.
\end{array}
\]
where $\alpha$ is determined such that $\alpha\underline{\pi}=1-\underline{\pi}$ with a given $\underline{\pi}\in (0,1)$.
The numbers $0$ and $1$ mean a nonaggressive action and an aggressive action, respectively.
The utility functions infer that the malicious sender prefers the aggressive action and the nonaggressive action when the receiver chooses the nonaggressive action and the aggressive action, respectively, and the receiver chooses the aggressive reaction when the sender is doubtful.
Specifically, the malicious sender's best responses are given by $a=0$ if $r=1$ and $a=1$ if $r=0$.
The receiver's best response is given by $r=1$ when $\pi_k(\htm)\geq \underline{\pi}$ and otherwise $r=0$.

In Fig.~\ref{fig:num}, the curves represent the malicious sender's estimated belief, two sample paths of the receiver's exact belief, and the empirical mean of the exact belief obtained under twenty trials where the broken line represents the boundary $\underline{\pi}$.
The curves are obtained under $\lambda=0.55$ and $\underline{\pi}=0.85$ with the initial belief $\pi_0(\htm)=0.15$.
Because the sample path of $\hat{\pi}_k(\htm)$ is independent of the elementary event under this setting, only the unique path is drawn.
The receiver's belief in the result is calculated with the true sender's best response.
As observed in the figure, the curve of the estimated belief $\hat{\pi}_k(\htm)$ monotonically increases as the time proceeds in contrast to the exact beliefs $\pi^{(1)}_k(\htm)$ and $\pi^{(2)}_k(\htm)$.
When the estimated belief exceeds $\underline{\pi}$, the malicious sender's best response becomes the benign one and subsequently the estimated belief is kept to be invariant.
Moreover, it can be confirmed that the analytic mean $\hat{\pi}(\tm)$ is close with the empirical mean, represented by the dotted line.
The result numerically indicates the asymptotic security of the control system guaranteed in the theoretical findings.

\begin{figure}[t]
\centering
\includegraphics[width = .98\linewidth]{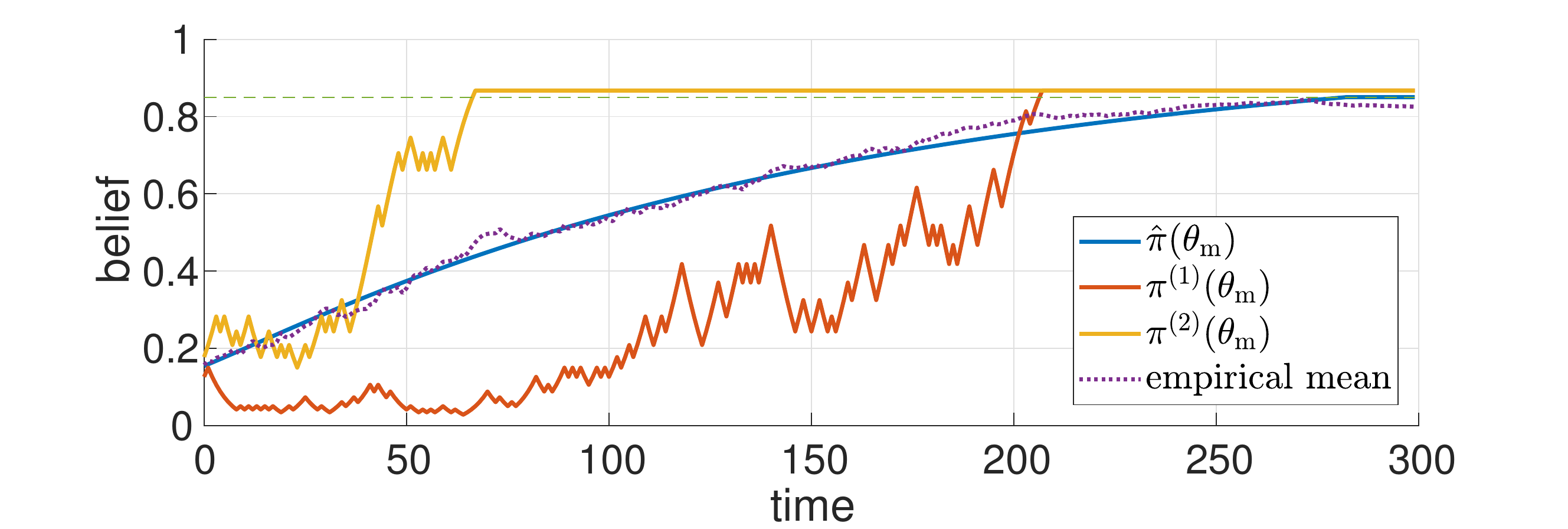}
\caption{The time sequence of the sender's estimated belief $\hat{\pi}_k(\htm)$, two sample paths of the exact receiver's belief $\pi(\htm)$, and the empirical mean.}
\label{fig:num}
\end{figure}

\section{Conclusion}
\label{sec:conc}
In this paper, we have considered a defense mechanism for resilient control systems with reactions.
The behaviors of a malicious attacker and the defender have been analyzed through a repeated signaling game.
It has theoretically been discovered that asymptotic security is achieved by covert reaction.
The result indicates the importance of covertness of reactions for enhancing resilience of control systems.

One important future research direction is to consider cumulative utilities and strategies with memory instead of the immediate utilities used for deciding the players' best responses.
In the current setting, because the sender does not care about future transition of receiver's belief, the malicious sender tends to choose an aggressive action even if it leads to a rapid rise of the belief, which could be unreasonable in practical situations.
Another direction is to consider the case where there is a feedback from the reaction to the system state.
It is expected that the same results hold if we can make the effects secret to the sender by utilizing the internal structure of the control system.
Moreover, analyzing ``transient security'' in the proposed framework is an open problem.
The difficulty arises from computational complexity of best responses for general control systems.
Finally, this paper is built on the premise that the sender uses her own strategy as the strategy estimated by the receiver for truncating the infinite chain of estimations.
Analysis of the case where multiple estimations are used for determining reasonable strategies is included in future work.

\section*{Appendix}
\subsection*{Proof of Lemma~\ref{lem:mono}}
Fix an elementary event and consider corresponding realizations.
All probabilities in the proof are taken with $A_k=\ssk(\tm,I^{\rm s}_k)$.
Since $\{\ssk\}$ is a sender's best response strategy profile, with $\sigszk = \sszk$ we have
$\hat{\pi}_{k+1}(\htm) = \displaystyle{\sum_{i^{\rm r}_{k} \in \mathcal{I}^{\rm r}_{k}} f(\pi_{k},i^{\rm r}_{k},\sszk)(\htm) p(i^{\rm r}_{k}|\tm,i^{\rm s}_{k})}.$
Since $p(\irk|i^{{\rm r}\prime}_{k-1})=0$ if $i^{{\rm r}\prime}_{k-1} \not\subset \irk$,
we have
\[
\begin{array}{cl}
p(i^{\rm r}_{k}|\tm,i^{\rm s}_{k}) \hspace{-1mm} \hs = \displaystyle{\sum_{i^{{\rm r}\prime}_{k-1} \in \mathcal{I}^{{\rm r}\prime}_{k-1}} p(\irk|\tm,\isk,i^{{\rm r}\prime}_{k-1})p(i^{{\rm r}\prime}_{k-1}|\tm,i^{\rm s}_{k})} \\
\hs = p(\irk|\tm,\isk,i^{\rm r}_{k-1})p(i^{\rm r}_{k-1}|\tm,i^{\rm s}_{k})
\end{array}
\]
where $i^{\rm r}_{k-1} \subset \irk$.
Because $\{U_k\}_{k\in\mathbb{Z}_+}$ are mutually independent and $a_{k-1}$ is uniquely determined from $i^{\rm s}_{k-1}$ with the given sender's strategy profile,
it turns out that
$p(i^{\rm r}_{k-1}|\tm,i^{\rm s}_{k}) = p(i^{\rm r}_{k-1}|\tm,i^{\rm s}_{k-1}, u_k, a_{k-1}) = p(i^{\rm r}_{k-1}|\tm,i^{\rm s}_{k-1}),$
which yields
\[
 \begin{array}{cl}
 \hat{\pi}_{k+1}(\htm) \hs = \displaystyle{\sum_{i^{\rm r}_{k} \in \mathcal{I}^{\rm r}_{k}} \dfrac{ p\left(\irk|\htm,i^{\rm r}_{k-1}\right)\pi_k(\htm) }{\displaystyle{\sum_{\theta' \in \hiTheta} p\left(\irk|\theta',i^{\rm r}_{k-1}\right)\pi_k(\theta')}}}\\
 \hs \quad \quad \cdot p(\irk|\tm,\isk,i^{\rm r}_{k-1})p(i^{\rm r}_{k-1}|\tm,i^{\rm s}_{k-1}).
 \end{array}
\]

We will reduce the index variable of the summation.
Because the control system is a deterministic system, $\hat{x}_k$ is uniquely determined once $u_{0:k}$ is fixed.
Hence, from Assumption~\ref{assum:info}, if $u_k,\hat{x}_k \notin \irk$ for $u_k\in\isk$ and $\hat{x}_k=\Sigma_k(u_{0:k},a_{\rm b})$, then $p(\irk|\theta,\isk,i^{\rm r}_{k-1})=0$.
Thus we have
\[
 \begin{array}{l}
 \hat{\pi}_{k+1}(\htm)\\
 = \displaystyle{\sum_{y_k\in\mathcal{Y}} \sum_{i^{\rm r}_{k-1}\in\mathcal{I}^{\rm r}_{k-1}} \dfrac{ p\left( y_k |\htm,i^{\rm r}_{k-1},u_k,\hat{x}_k\right)\pi_k(\htm) }{\displaystyle{\sum_{\theta' \in \hiTheta} p\left(y_k|\theta',i^{\rm r}_{k-1},u_k,\hat{x}_k,\right)\pi_k(\theta')}}}\\
 \quad \cdot p(y_k|\tm,\isk,i^{\rm r}_{k-1})p(i^{\rm r}_{k-1}|\tm,i^{\rm s}_{k-1}).
 \end{array}
\]
Since the estimated state $\hat{x}_k$ is uniquely determined from $i^{\rm r}_{k-1}$,
we have $p\left( y_k |\htm,i^{\rm r}_{k-1},u_k,\hat{x}_k\right) = p(y_k|\tm,\isk,i^{\rm r}_{k-1}),$
which we denote by $p_{\rm m}(y_k,i^{\rm r}_{k-1})$.
With the similar notation $p_{\rm b}(y_k,i^{\rm r}_{k-1}):=p(y_k|\tb,\isk,i^{\rm r}_{k-1})$, we have
\[
 \begin{array}{l}
 \hat{\pi}_{k+1}(\htm)\\
 \hspace{-1mm} = \displaystyle{\sum_{y_k\in\mathcal{Y}} \sum_{i\in\mathcal{I}^{\rm r}_{k-1}} \dfrac{ p_{\rm m}(y_k,i)^2 \pi_k(\htm)p(i|\tm,i^{\rm s}_{k-1}) }{\displaystyle{ p_{\rm b}(y_k,i) \pi_k(\htb)} + p_{\rm m}(y_k,i) \pi_k(\htm)}}
 \end{array}
\]
where $i$ is used instead of $i^{\rm r}_{k-1}$ for simplifying the expression.
The equation can be rewritten as
\begin{equation}\label{eq:Gi}
 \hat{\pi}_{k+1}(\htm) =\hspace{-3mm} \sum_{i\in\mathcal{I}^{\rm r}_{k-1}} G_k(i) \pi_k(\htm)p(i|\tm,i^{\rm s}_{k-1})
\end{equation}
where
$G_k(i) := \sum_{y\in\mathcal{Y}} \dfrac{ p_{\rm m}(y_k,i)^2 }{\displaystyle{ p_{\rm b}(y_k,i) \pi_k(\htb)} + p_{\rm m}(y_k,i) \pi_k(\htm)}.$

We investigate the exact lower bound of $G_k(i)$ for a fixed $i$.
Omitting the notation $i$ and $k$, we let $m_y:=p_{\rm m}(y_k,i)$ and $b_y:=p_{\rm b}(y_k,i)$, which satisfy $\sum_{y\in\mathcal{Y}}m_y=\sum_{y\in\mathcal{Y}}b_y=1$.
Define $g(\alpha,m,b):= \sum_{y\in \mathcal{Y}} \frac{m_y^2}{\alpha m_y + (1-\alpha) b_y}$
where $m$ and $b$ contain all $m_y$ and $b_y$, respectively, and $\alpha \in (0,1)$.
Since
\[
 \dfrac{m_y^2}{\alpha m_y + (1-\alpha) b_y} = \dfrac{m_y}{\alpha}-\dfrac{1-\alpha}{\alpha}\cdot \dfrac{m_y b_y}{\alpha m_y + (1-\alpha)b_y},
\]
we have
\[
 g(\alpha,m,b) = \dfrac{1}{\alpha} \sum_{y\in \mathcal{Y}}m_y - \dfrac{1-\alpha}{\alpha} \cdot \underbrace{\sum_{y\in\mathcal{Y}} \dfrac{m_y b_y}{\alpha m_y + (1-\alpha)b_y}}_{=: h(\alpha,m,b)}.
\]
From the weighted arithmetic mean--geometric mean inequality, it turns out that
\[
 \begin{array}{cl}
 h(\alpha,m,b) \hs \leq \sum_{y\in \mathcal{Y}} \dfrac{m_y b_y}{m_y^{\alpha} b_y^{1-\alpha}} = \sum_{y\in \mathcal{Y}} m_y^{1-\alpha}b_y^{\alpha}\\
  \hs \leq \sum_{y \in \mathcal{Y}} \left\{(1-\alpha)m_y+\alpha b_y\right\} = 1.
 \end{array}
\]
By using this inequality, we have
$g(\alpha,m,b) \geq \frac{1}{\alpha} \sum_{y\in \mathcal{Y}}m_y - \frac{1-\alpha}{\alpha} = 1.$
Since $G_k(i)=g(\pi_k(\tb),m,b)$, we have $G_k(i) \geq 1$, where the equality holds if and only if $m_y=b_y$ for any $y\in\mathcal{Y}$.

From the lower bound of $G_k(i)$, we have
$\hat{\pi}_{k+1}(\htm) \geq \sum_{i\in\mathcal{I}^{\rm r}_{k-1}} \pi_k(\htm)p(i|\tm,i^{\rm s}_{k-1}) = \hat{\pi}_k(\htm)$
where the equality holds if and only if
$p(y_k|\tb,\isk,i^{\rm r}_{k-1})=p(y_k|\tm,\isk,i^{\rm r}_{k-1})$
for any $y_k\in \mathcal{Y},i^{\rm r}_{k-1} \in \mathcal{I}^{\rm r}_{k-1}$.
Since
$p(y_k|\theta,\isk,i^{\rm r}_{k-1})= \lambda_k(y_k|\Sigma_k(u_{0:k},s^{\rm s}_k(\theta,i^{\rm s}_k))),$
the positivity of the mutual information in Assumption~\ref{assum:defense} yields
$\Sigma_k(u_{0:k},s^{\rm s}_k(\tb,i^{\rm s}_k)) = \Sigma_k(u_{0:k},s^{\rm s}_k(\tm,i^{\rm s}_k)).$
Thus, from the input observability in Assumption~\ref{assum:defense},
$s^{\rm s}_k(\tb,i^{\rm s}_k)=s^{\rm s}_k(\tm,i^{\rm s}_k).$
From Proposition~\ref{prop:benign}, the claim holds.
\hfill\QED

\subsection*{Proof of Theorem~\ref{thm:main}}
Fix an elementary event in $\mathcal{E}_{\{\ssk\}}$ and consider corresponding realizations.
Because $\{\hat{\pi}_k(\htm)\}_{k\in\mathbb{Z}_+}$ is bounded above, Lemma~\ref{lem:mono} implies that there exists $\pia \in [0,1]$ such that $\hat{\pi}_k(\htm)\to\pia$ as $k\to \infty$.

We show that there exists $N\in\mathbb{Z}_+$ such that $\hat{\pi}_{k}(\htm)$ is identical for $k>N$ by contradiction.
Assume that $\hat{\pi}_{k+1}(\htm)>\hat{\pi}_{k}(\htm)$ for any $k\in\mathbb{Z}_+$.
Then for any $\epsilon>0$ there exists $N\in \mathbb{Z}_+$ such that $\pia-\epsilon<\hat{\pi}_k(\htm)<\pia$ for any $k>N$.
From the uniform boundedness of the mutual information in Assumption~\ref{assum:defense},
$\limsup_{k\to \infty} \max_{y\in \mathcal{Y}, i\in \mathcal{I}^{\rm r}_{k-1}} |p_{\rm m}(y_k,i)-p_{\rm b}(y_k,i)|>0.$
Because $G_k(i)$, defined in~\eqref{eq:Gi}, is a continuous function with respect to $p_{\rm m}(y_k,i)$ and $p_{\rm b}(y_k,i)$ and $\hat{\pi}^{\ast}_{\rm m}<1$ from Assumption~\ref{assum:prefe}, we have
$\limsup_{k\to \infty} \max_{i\in \mathcal{I}^{\rm r}_{k-1}} G_k(i) >1.$
Then there exists $\epsilon>0$ and $k\in \mathbb{Z}_+$ such that $\hat{\pi}_{k+1}(\htm)>\pia$ holds when $\pia-\epsilon<\hat{\pi}_k(\htm)$, which leads to a contradiction.

Therefore, there exists $N\in\mathbb{Z}_+$ such that $\hat{\pi}_{k}(\htm)$ is identical for $k>N$.
Hence, from Lemma~\ref{lem:mono}, $\ssk(\tm,\isk)=a_{\rm b}$ for $k>N$, which proves the claim.
\hfill\QED

\bibliographystyle{IEEEtran}
\bibliography{sshrrefs}


\end{document}